\documentclass[12pt]{iopart}

\usepackage{graphicx}

\usepackage{subfig}

\usepackage{iopams}  
\usepackage{bbm}
\usepackage{amstext}
\usepackage{amsthm}
\usepackage{multirow}

\newcommand{\beq}{\begin{equation}}
\newcommand{\eeq}{\end{equation}}
\newcommand{\ket}[1]{|#1\rangle}
\newcommand{\bra}[1]{\langle#1|}

\newcommand{\cL}{\mathcal{L}}
\newcommand{\cQ}{\mathcal{Q}}
\newcommand{\I}{\mathbbm{1}}
\newcommand{\eqref}[1]{(\ref{#1})}

\newtheorem{theorem}{Theorem}
\newtheorem{lemma}[theorem]{Lemma}

\newtheorem{proposition}[theorem]{Proposition}

\newtheorem{conjecture}[theorem]{Conjecture}

\begin{document}

\title[Quantifying non-classical and beyond-quantum correlations]{Quantifying non-classical and beyond-quantum correlations in the unified operator formalism}

\author{Joshua Geller and Marco Piani}
\address{Institute for Quantum Computing and Department of Physics and Astronomy, University of Waterloo, N2L 3G1 Waterloo ON, Canada}

\ead{mpiani@uwaterloo.ca}
\begin{abstract}
Acin et al.~\cite{acin} introduced a unified framework for the study of no-signalling correlations. Such a framework is based on the notion of local quantum measurements, but, in order to account for beyond-quantum correlations, global pseudo-states that are not positive semidefinite are allowed. After a short review of the formalism, we consider its use in the quantification of both general non-local and beyond-quantum correlations. We argue that the unified framework for correlations provides a simple approach to such a quantification, in particular when the quantification is meant to be operational and meaningful in a resource-theory scenario, i.e., when considering the processing of resources by means of non-resources. We relate different notions of robustness of correlations, both at the level of \mbox{(pseudo-)states} and abstract probability distributions, with particular focus on the beyond-quantum robustness of correlations and pseudo-states. We revisit known results and argue that, within the unified framework, the relation between the two levels---that of operators and that of probability distributions---is very strict. We point out how the consideration of robustness at the two levels leads to a natural framework for the quantification of entanglement in a device-independent way. Finally, we show that the beyond-quantum robustness of the non-positive operators needed to achieve beyond-quantum correlations coincides with their negativity and their distance from the set of quantum states. As an example, we calculate the beyond-quantum robustness for the case of a noisy Popescu-Rohrlich box. 
\end{abstract}

\maketitle

\section{Introduction}
Bell~\cite{bell} proved that correlations that arise classically cannot reproduce some of the correlations allowed by quantum mechanics~\cite{nonlocalityreview}. One could say that quantum correlations require enlarging our view of what correlations can exist between distant systems. Taking a slightly different point of view, one can interpret this as a \emph{limitation} of classical correlations: they are not as strong as instead possible in the quantum setting, which we believe to constitute the real physical setting of our world.

Quantum correlations, albeit stronger than classical, are not arbitrary~\cite{tsirelson}. Most importantly, they respect the so-called no-signalling principle, i.e., they cannot lead to instantaneous communication. While one can argue that the no-signalling condition is
somewhat built-in in the quantum formalism,
 the past has seen attempts to find ways to use non-local correlations, arising from \emph{entanglement}~\cite{entanglementreview}, for faster-than-light communication~\cite{herbert}. The refutation of such attempts has been instrumental in the establishment of the no-cloning theorem~\cite{wootterszurek,dieks} and in the development of quantum information processing and of quantum cryptography~\cite{NC}. It is then fair to say that the study of the limits of quantum correlations themselves has been extremely fruitful.

Quite interestingly, Popescu and Rohrlich~\cite{PR} provided examples of---unphysical, it is believed---correlations that, although no-signalling, go beyond what possible quantumly. Since that seminal paper, a lot of work has been devoted to the attempt to ``pin down'' quantum correlations,
often assuming as larger ``potential'' set of allowed correlations than that of no-signalling correlations. From a foundational standpoint, the effort is mostly directed to finding some operational and/or information-theoretic principle that, together with, or extending, the no-signallling principle, would uniquely identify quantum correlations among general no-signalling ones (see, e.g.,~\cite{informationcausality} and~\cite{localorthogonality}). From the above remarks it should be clear that this is a worthwhile effort  not only  from a foundational perspective: a better understanding of quantum correlations (and of their limits) can have enormous repercussions from the applicative point of view. One example of this is the development of device-independent quantum cryptography~(see~\cite{varizani} and references therein).

Acin et al.~\cite{acin} (see also~\cite{short}) introduced a unified formalism to describe all no-signalling correlations, even beyond quantum, in, we could say,  ``(at least partially) familiar quantum terms''. Acin et al.'s formalism is based on the notion of local quantum measurements, but, in order to account for beyond-quantum correlations, global pseudo-states that are not positive semidefinite are allowed. In this paper we show how that formalism can be used to understand and quantify correlations in a unified way, also from a resource perspective.

\section{Preliminaries}

\subsection{Correlations}

We will study correlations in a bipartite setting, involving the two parties Alice and Bob, but many of the considerations and results can be easily extended to the multipartite scenario. Correlations will be described abstractly in terms of conditional probability distributions for the choice of a local measurements and the recording of the outcomes of such measurements. Measurement choices will be denoted by $x$ and $y$ for Alice and Bob, with outcomes labelled by $a$ and $b$, respectively. Both measurement choices and possible outcomes range within definite alphabets, e.g., $x\in\{0,1,\ldots,|X|-1\}$, defining some input and output dimensions, e.g., $|X|$ for the number of choices of measurement for Alice.  In this paper we will not be concerned with such dimensions, which are nonetheless considered fixed for a given ``experiment''.

In this framework, the correlations that exist between Alice and Bob are described by means of an ordered collection (a vector)  $\vec{p}$ of numbers $p(a,b|x,y)$, each giving the probability of seeing the pair of outcomes $(a,b)$ upon the choice of the measurement pair $(x,y)$. Thus, these numbers are all non-negative, $p(a,b|x,y)\geq0$, and satisfy the normalization $\sum_{a,b}p(a,b|x,y)=1$ for all $(x,y)$, from which also the bound $p(a,b|x,y)\leq1$ follows.

We have already used the terms ``input'' and ``output'', and indeed we will think of $\vec{p}$ as of a box accessed independently by the two parties---users---Alice and Bob, as depicted in Fig.~\ref{fig:box}(a).
\begin{figure}
\begin{center}
\subfloat[Correlations as boxes]{
		\includegraphics[scale=0.6]{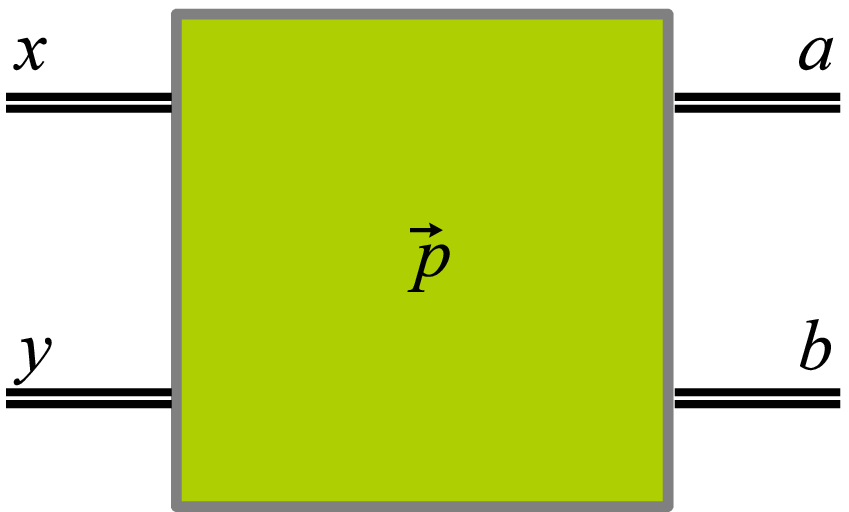}}
\hspace{0.5cm}
\subfloat[$O$-formalism]{
		\includegraphics[scale=0.6]{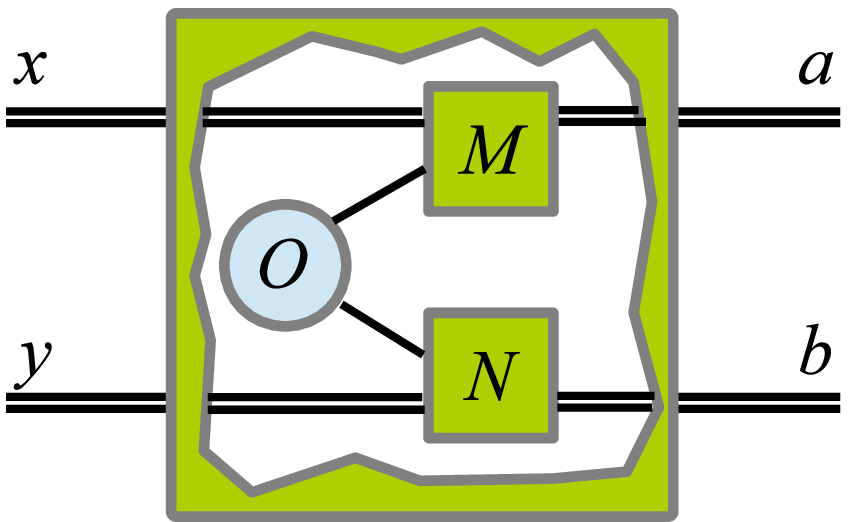}}
\end{center}
\caption{(a) Correlations are abstractly described as boxes with inputs (the choice of local measurements) and outputs (the outcomes of the measurements). Such inputs/outputs are classical (i.e., can be described and transmitted by classical bits) and are represented by double lines. Information flows from left to right. (b) All no-signalling correlations/boxes (even those the are not quantum) can be realized via local measurements acting on a state (a normalized positive semidefinite operator) or a pseudo-state (that is, a trace-one Hermitian operator that is not positive semidefinite) $O$~\cite{acin}: it always possible to ``open the box'' to reveal its (pseudo-)quantum implementation. Single lines depict the transmission/flow of quantum information, including the case of pseudo-states.}
\label{fig:box}
\end{figure}
For any fixed number of possible inputs/ouputs, we can consider the space of all allowed $\vec{p}$, and introduce a meaningful, physically-motivated classification of these probability vectors~\footnote{Each $\vec{p}$ is actually a collection of \emph{conditional} probability vectors, but for the sake of brevity we may refer to it as probability vector.}.

\subsubsection{{\bf Local correlations}} Local correlations are those that can be expressed as
\beq
\label{eq:localp}
p_{AB}(a,b|x,y)=\sum_i q_i p^i_A(a|x)p^i_B(b|y)
\eeq
where $p^i_A(a|x)$ are \emph{local} conditional probability distributions for Alice (similarly for Bob) that depend also on $i$, and $q_i$ are probabilities, i.e., $q_i\geq 0$ and $\sum_i q_i=1$. That is, local correlations are convex combinations of uncorrelated local probabilities distributions. With the use of subscripts, we have emphasized here the corresponding parties; such subscripts will be omitted whenever confusion does not arise. Notice that in vector notation we could write $\vec{p}_{AB}= \sum_iq_i \vec{p^i}_A\otimes \vec{p^i}_B$. We will denote the space of all local correlations (for a fixed choice of number of inputs/outputs) by $\mathcal{L}$. If a probability $\vec{p}$ does not admit the decomposition (\ref{eq:localp}), then we call it \emph{non-local}.

\subsubsection{{\bf Quantum correlations}} Quantum correlations are those that can be expressed as
\beq
\label{eq:quantump}
p_{AB}(a,b|x,y)=\Tr(M^A_{a|x}\otimes N^B_{b|y} \rho_{AB}),
\eeq
where $\rho_{AB}$ is a distributed quantum state,
and $M^A_{a|x}$ and $N^B_{b|y}$ are the elements of  positive-operator-valued measures (POVMs), one for each choice of measurement: $M^A_{a|x}\geq 0$ $\forall a,x$, $\sum_a M^A_{a|x} =\mathbbm{1}_A$ $\forall x$, and, similarly $N^B_{b|y}\geq 0$ $\forall b,y$, $\sum_b N^B_{b|y} =\mathbbm{1}_B$ $\forall y$. Since it will be important in the following, we remark that $\rho_{AB}$ is a quantum state if and only if $\rho_{AB}\geq 0$ an $\Tr(\rho_{AB})=1$. We will denote the space of all quantum correlations (for a fixed choice of number of inputs/outputs) by $\mathcal{Q}$. Notice that, even if the number of inputs/outputs is fixed, we allow the realization of the box (i.e., probability distribution) with quantum systems of arbitrary dimensions (correspondingly, $\mathbbm{1}_A$ is the identity operator on the ``underlying'' arbitrary quantum system of Alice).

\subsubsection{{\bf No-signalling correlations}} No-signalling correlations are those that respect
\begin{eqnarray}
\sum_a p(a,b|x,y) &= \sum_a p(a,b|x',y)\quad &\forall x,x',\\
\sum_b p(a,b|x,y) &= \sum_b p(a,b|x,y')\quad &\forall y,y'.
\end{eqnarray}
These conditions allow one to define local conditional probabilities distributions, e.g., $p_A(a|x)=\sum_b p(a,b|x,y)$. Notice that in absence of the no-signalling conditions, there should be in general a dependence on $y$ on the left-hand side of the last equation. The lack of such a dependence is exactly what constitutes no-signalling: even if Bob changes his input, he cannot send a message to Alice by modifying the probability distributions she sees locally. We will denote the space of all no-signalling correlations (for a fixed choice of number of inputs/outputs) by $\mathcal{NS}$.

We observe that both local and quantum correlations are no-signalling. This follows immediately by inspection, given (\ref{eq:localp}) and (\ref{eq:quantump}). Notice that in the quantum case this is built-in in the formalism thanks to the fact that, e.g., $\sum_b N^B_{b|y} =\mathbbm{1}_B$ $\forall y$, with the right-hand side of the equality independent of $y$. On the other hand, classical correlations are a subset of quantum correlations (with strict containment, as Bell first showed~\cite{nonlocalityreview}), so overall we have (see Fig.~\ref{fig:box}b)
\begin{figure}
\begin{center}
\includegraphics[scale=0.6]{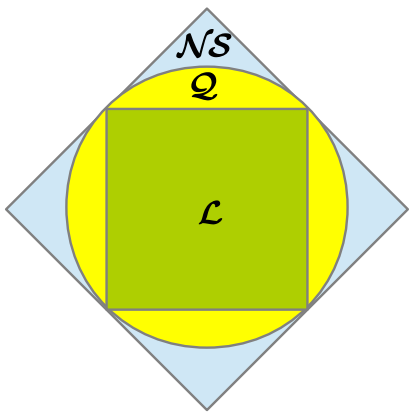}
\caption{Hierarchy of correlations. Local correlations ($\mathcal{L}$) are a strict subset of quantum correlations ($\mathcal{Q}$), which in turn are a strict subset of no-signalling correlations ($\mathcal{NS}$). The figure depicts also the fact that $\mathcal{L}$ and $\mathcal{NS}$ are polytopes; this is not the case for $\mathcal{Q}$~\cite{nonlocalityreview}.}
\label{fig:sets}
\end{center}
\end{figure}
\beq
\mathcal{L}\subsetneq \mathcal{Q} \subsetneq \mathcal{NS}.
\eeq
The last is also a strict containment, as proven by Popescu and Rohrlich~\cite{PR} (see Section~\ref{sec:PR}).

\subsection{Separability and entanglement}
Bipartite quantum states can be classified in terms of their bipartite structure. We call \emph{separable} (or \emph{unentangled}) the states that can be expressed as convex combination of local product states, i.e.,
\beq
\sigma_{AB} = \sum_i p_i \sigma_A^i\otimes \sigma_B^i
\eeq
and \emph{entangled} the states that do not admit such a decomposition. Werner~\cite{werner} introduced this general notion for mixed-state entanglement, and showed that the relation between the notions of locality/non-locality and separability/entanglement is at best non-trivial. On one hand, it is easy to verify that every separable state gives raise only to local correlations, independently of the local measurements.
 Also, it is easy to check that every local $\vec{p}$ admits a realization (via (\ref{eq:quantump}))  by means of some separable state. On the other hand, quite more surprisingly, there are entangled states that  via (\ref{eq:quantump}) give only raise to local correlations~\cite{werner}. This fact has led to a number of attempts to find a more ``faithful'' relation between quantum states and correlations, or, more in general, between quantum states and non-local features of correlations, going beyond (\ref{eq:quantump}), so to potentially reveal the ``hidden'' non-locality of all entangled states~\cite{nonlocalityreview}.

We will denote by $\mathbb{D}$ the set of all density matrices, and by $\mathbb{S}$ the set of all \emph{separable} density matrices. As clear from what we already said, $\mathbb{S}\subsetneq\mathbb{D}$.

\subsection{Unified framework for correlations}

Equation~(\ref{eq:quantump}) establishes a (not one-to-one) two-way mapping between quantum states and quantum correlations, thus including local correlations. In~\cite{acin} (see also~\cite{short}) a unified framework for correlations (valid also in the multipartite setting) very similar to (\ref{eq:quantump}) was introduced for \emph{all} no-signalling correlations. It reads (see Figure~\ref{fig:box}(a))
\beq
\label{eq:oformalism}
p_{AB}(a,b|x,y)=\Tr(M^A_{a|x}\otimes N^B_{b|y} O_{AB}),
\eeq
where we still use local POVMs, but the only request on the Hermitian operator $O_{AB}$, which we could call \emph{pseudo-state}, is that it is normalized, $\Tr(O_{AB})=1$, and that, together with a---possibly \emph{very specific}~\footnote{A pseudo-state that is not positive semidefinite may lead to non-positive pseudo-probability distributions if the POVMs entering \eqref{eq:oformalism} are chosen arbitrarily.}---choice of POVMs, gives raise to \emph{bona fide} correlations, i.e., conditional probabilities that are positive---normalization is ensured by the fact that  $O_{AB}$ has unit trace. In~\cite{acin} it was proven that a (\emph{bona-fide}) $\vec{p}$ is no-signalling if and only if it is possible to choose a pseudo-state $O$ and local POVMs such that (\ref{eq:oformalism}) is satisfied. This means that also a no-signalling $\vec{p}$ that is not quantum can be represented within the quantum formalism at the cost of dealing with ``non-positive states''. This opens up the possibility of looking at general correlations (even beyond quantum) in a novel, unified way that may shed light on the properties of quantum correlations themselves. We remark that the proof presented in \cite{acin} is constructive: for a given no-signalling $\vec{p}$, a recipe is given to find a specific choice of measurements and of $O$ such that (\ref{eq:oformalism}) is satisfied. In general, we write $O\rightarrow \vec{p}$ to indicate that there exist local measurements such that (\ref{eq:oformalism}) holds. In the case we want to specify which particular choice of $O$ and of local measurements $\{M_{a|x}\}$, $\{N_{b|y}\}$ leads to $\vec{p}$, we write $(O, \{M_{a|x}\}, \{N_{b|y}\})\rightarrow \vec{p}$.

We will denote by $\mathbb{O}$ the set of all bipartite normalized (i.e., of unit trace) Hermitian operators, i.e., pseudo-states. Obviously $\mathbb{S}\subsetneq\mathbb{D}\subsetneq\mathbb{O}$.

\section{Robustness measures of correlations}

In this paper we focus on the issue of quantification of correlations.  We will adopt a resource-theory view that aims 
to quantitatively differentiate between the various kinds of correlations that we have introduced: general no-signalling, quantum and local.

\subsection{Robustness}

Consider a real affine space $\mathcal{A}$ such that $(1-a) v_1+ a v_2\in \mathcal{A}$ for any $v_1,v_2\in\mathcal{A}$, and for any $a\in\mathbb{R}$. The robustness of an element $v\in\mathcal{A}$ with respect to a closed convex subset $S$ of $\mathcal{A}$ that spans the latter, can be generically defined as
\beq
\label{eq:robustness}
r_S(v):=\min\left\{t \left| t\geq0, \frac{v+t w}{1+t}\in S \textrm{ for some } w\in S \right\}\right.
\eeq
We can define $r_S(v)$ as a minimum because we suppose that  $S$ is convex and closed, and spans $\mathcal{A}$ (via affine combinations), so that  $r_S(v)<+\infty$ is achieved for all $v\in\mathcal{A}$. On the other hand, one has $r_S(v)=0$ if and only if $v\in S$. Further,
from the definition it follows that
it is possible to express $v$ as an \emph{affine} combination $v = (1+r_S(v)) w^+ - r_S(v) w^-$ of two elements $w^\pm$ of $S$ (see Fig.~\ref{fig:robustness}).
\begin{figure}
\begin{center}
\includegraphics[scale=0.6]{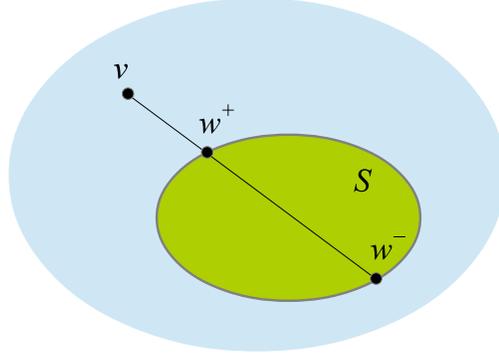}
\caption{Robustness of $v$ with respect to the set convex and closed set $S$; $v$ can be expressed as an affine combination of $w^\pm\in S$.}
\label{fig:robustness}
\end{center}
\end{figure}
Robustness expresses how much ``noise''---in terms of random mixture with an element of $S$, $w^-$---is enough to ``move'' $v$ into the set $S$, making it become $w^+$.

It is clear that in this sense, the robustness is a quantifier of ``how far'' from being part of $S$ the element $v$ is. In our setting, one can define several ``robustnesses'', both for non-signalling bipartite correlations and for (pseudo-)states:
\begin{eqnarray}
r_\mathcal{L}(\vec{p})&:=\min\left\{t \left|\; t\geq0, \frac{\vec{p}+t \vec{p}_\mathcal{L}}{1+t}\in \mathcal{L} \textrm{ for some } \vec{p}_\mathcal{L} \in \mathcal{L} \right\}\right.,\\
r_\mathcal{Q}(\vec{p})&:=\min\left\{t \left|\; t\geq0, \frac{\vec{p}+t \vec{p}_\mathcal{Q}}{1+t}\in \mathcal{Q} \textrm{ for some } \vec{p}_\mathcal{Q} \in \mathcal{Q} \right\}\right.,\\
\label{eq:entanglementrobustness}r_\mathbb{S}(O)&:=\min\left\{t \left|\; t\geq0, \frac{O+t \sigma_\mathbb{S}}{1+t}\in \mathbb{S} \textrm{ for some } \sigma_\mathbb{S} \in \mathbb{S} \right\}\right.,\\
\label{eq:pseudorobustness}r_\mathbb{D}(O)&:=\min\left\{t \left|\; t\geq0, \frac{O+t \rho_\mathbb{D}}{1+t}\in \mathbb{D} \textrm{ for some } \rho_\mathbb{D} \in \mathbb{D} \right\}\right..
\end{eqnarray}
We remark that we used the fact that local correlations span $\mathcal{NS}$~\cite{communicationcomplexity} and that separable states span the set of bipartite states~\cite{separabilityball}, and hence the set of pseudo-states.
The quantities $r_\mathcal{L}$ and $r_\mathcal{Q}$, respecting  $r_\mathcal{L}\geq r_\mathcal{Q}$, have already been considered in, for example,~\cite{communicationcomplexity}. The quantity $r_\mathbb{S}$ was originally defined in~\cite{robustness} for bipartite states, and is known as \emph{robustness} of entanglement: here we extend it to pseudo-states. Finally, $r_\mathbb{D}(O)$ has appeared previously in literature in the context of the study of negativity of entanglement~\cite{negativity}. Clearly, $r_\mathbb{S}\geq r_\mathbb{D}$.

\subsection{Robustness in  the unified operator formalism}

Given the two-way mapping of (\ref{eq:oformalism}) one can also introduce ``hybrid'', less trivial notions of robustness, not directly fitting in the framework of~(\ref{eq:robustness}). For example, for fixed $\vec{p}$ we can define
\beq
\label{eq:mixedrobustness}
r_\mathbb{S}(\vec{p}):=\inf_{O\rightarrow \vec{p}} r_\mathbb{S}(O),\qquad
r_\mathbb{D}(\vec{p}):=\inf_{O\rightarrow \vec{p}} r_\mathbb{D}(O),
\eeq
where the infima are taken over all $O$ such that $O\rightarrow \vec{p}$. Notice that we revert to considering infima because we are not constraining the underlying dimension of $O$, but we will see that the infima are actually minima. Also, notice that the quantities are well defined because of the construction of~\cite{acin}, which ensures that there is at least one $O$ such that  $O\rightarrow \vec{p}$. The two measures $r_\mathbb{S}(\vec{p})$ and $r_\mathbb{D}(\vec{p})$ quantify how far from separable and positive semidefinite, respectively, any pseudo-state $O$ needs  to be at least in order to give raise to $\vec{p}$ via some local measurements. On the other hand, for fixed $O$ we can define
\beq
r_\mathcal{L}(O):=\sup_{O\rightarrow \vec{p}}r_\mathcal{L}(\vec{p}),\qquad
r_\mathcal{Q}(O):=\sup_{O\rightarrow \vec{p}}r_\mathcal{Q}(\vec{p}),
\eeq
where the suprema are over all possible local measurements, with an arbitrary number of inputs and outputs---hence the use of the suprema, rather than maxima. These latter  two measures quantify how non-local (quantum) a correlation vector $\vec{p}$ generated via local measurements on $O$ can be.
Notice that if $O$ is not positive semidefinite or an entanglement witness~\cite{acin}, then there will necessarily be local measurements such that the resulting $\vec{p}$ is not a good probability vector, because some of the entries will be negative.  In such a case one can still calculate the robustness of the pseudo-probability distribution $\vec{p}$ with respect to $\mathcal{L}$ and $\mathcal{Q}$. Alternatively, one can potentially consider an additional restriction, imposing that only measurements leading to \emph{bona fide} probability distributions are allowed.

Suppose that $O$ is such that $O\rightarrow\vec{p}$. By definition of $r_\mathbb{S}(O)$, there exists a separable state $\sigma_\mathbb{S}$ such that $(O+ r_\mathbb{S}(O) \sigma_\mathbb{S})/(1+r_\mathbb{S}(O))$ is  separable. By locally measuring the latter state  with the same measurements that give the mapping from $O$ to $\vec{p}$, we  obtain a local probability distribution $(\vec{p}+ r_\mathbb{S}(O) \vec{p}_\mathcal{L})/(1+r_\mathbb{S}(O))$, with $ \vec{p}_\mathcal{L}\in \mathcal{L}$, because the mapping from operators to boxes---for fixed measurements---is linear, and local measurements performed on a separable state always give raise to a local probability distribution. Since this is true for any $O$ such that $O\rightarrow\vec{p}$, this proves that $r_\mathcal{L}(\vec{p})\leq \inf_{O\rightarrow \vec{p}} r_\mathbb{S}(O)=r_\mathbb{S}(\vec{p})$. One similarly proves that $r_\mathcal{Q}(\vec{p})\leq r_\mathbb{D}(\vec{p})$. We will now see that these are actually equalities. Before we proceed, we need a simple lemma.

\begin{lemma}\label{lem:Oadd}
Given two probability distributions $\vec{p}$ and $\vec{p'}$ such that $(O, \{M_{a|x}\}, \{N_{b|y}\}) \rightarrow \vec{p}$ and $(O', \{M'_{a|x}\}, \{N'_{b|y}\}) \rightarrow \vec{p'}$, then
\beq
\label{eq:convexprob}
(\tilde{O}, \{\tilde{M}_{a|x}\}, \{\tilde{N}_{b|y}\}) \rightarrow \vec{\tilde{p}} := (1-q)\vec{p} + q\vec{p'},
\eeq
with
\begin{eqnarray}
\nonumber{\tilde{O}}^{ABCD} &:= (1-q) O^{AB} \otimes \ket{0}\bra{0}^C\otimes\ket{0}\bra{0}^D + q O'^{AB} \otimes \ket{1}\bra{1}^C\otimes\ket{1}\bra{1}^D, \\
\label{eq:POVM_A}\tilde{M}^{AC}_{a|x} &:= M^{A}_{a|x} \otimes \ket{0}\bra{0}^{C} + M'^{A}_{a|x} \otimes \ket{1}\bra{1}^{C}, \\
\label{eq:POVM_B}\tilde{N}^{BD}_{b|x} &:= N^{B}_{b|y} \otimes \ket{0}\bra{0}^{D} + N'^{B}_{b|y} \otimes \ket{1}\bra{1}^{D},
\end{eqnarray}
for any $q \in \mathbb{R}$.
\end{lemma}
\begin{proof}
By inspection: it is easy to check both that $\{\tilde{M}^{AC}_{a|x}\}$ and $\{\tilde{N}^{BD}_{b|x}\}$ are \emph{bona fide} POVMs (for fixed $x$ and $y$, respectively), and that (\ref{eq:convexprob}) is satisfied.
\end{proof}

\begin{proposition}
\label{prop:equalityrobustness}
It holds $r_\mathcal{L}(\vec{p})= r_\mathbb{S}(\vec{p})$ and $r_\mathcal{Q}(\vec{p})= r_\mathbb{D}(\vec{p})$.
\end{proposition}
\begin{proof}
We already argued that $r_\mathcal{L}(\vec{p})\leq r_\mathbb{S}(\vec{p})$ and $r_\mathcal{Q}(\vec{p})\leq r_\mathbb{D}(\vec{p})$. 
We will now explicitly prove $r_\mathcal{L}(\vec{p})\geq r_\mathbb{S}(\vec{p})$; $r_\mathcal{Q}(\vec{p})\geq r_\mathbb{D}(\vec{p})$ can be  proven along the same lines.

By definition of $r_\mathcal{L}(\vec{p})$, it holds that there are  local $\vec{p}^\pm_\mathcal{L}$ such that
 $\vec{p}=(1+r_\mathcal{L}(\vec{p})\vec{p}^+_\mathcal{L} - r_\mathcal{L}(\vec{p})\vec{p}^-_\mathcal{L}$. Consider now \emph{separable} $O^{\pm}$ and POVMs $\{M^{\pm}_{a|x}\}, \{N^{\pm}_{b|y}\}$ such that $(O^\pm, \{M^\pm_{a|x}\}, \{N^\pm_{b|y}\})\rightarrow\vec{p}^\pm_\mathcal{L}$, which are known to exist~\cite{acin}.
 Then we can follow the construction of Lemma~\ref{lem:Oadd} to find an $\tilde{O}=(1+r_\mathcal{L}(\vec{p})) O^{+,AB} \otimes \ket{0}\bra{0}^C\otimes\ket{0}\bra{0}^D - r_\mathcal{L}(\vec{p})O^{-,AB} \otimes \ket{1}\bra{1}^C\otimes\ket{1}\bra{1}^D$ such that $\tilde{O}\rightarrow\vec{p}$. Since $O^{+,AB} \otimes \ket{0}\bra{0}^C\otimes\ket{0}\bra{0}^D$ and $O^{-,AB} \otimes \ket{1}\bra{1}^C\otimes\ket{1}\bra{1}^D$  are both separable in the $AC:BD$ cut, we see that  $r_\mathbb{S}(\vec{p})=\inf_{O\rightarrow \vec{p}} r_\mathbb{S}(O)\leq r_\mathcal{L}(\vec{p})$.
\end{proof}

The latter result means that the infima in \eqref{eq:mixedrobustness} are achieved, and can be taken to be minima. It is also worth remarking that, in the construction of $\tilde{O}$ in the proof of Proposition~\ref{prop:equalityrobustness}, it is not assured that $\tilde{O}$ is a quantum state, even in the case where $\vec{p}$ is quantum. Indeed,  in our definition of $r_\mathbb{S}(\vec{p})$ we do not presuppose anything about $O$; in particular, we do not assume that, in the case of a quantum $\vec{p}$, an $O$ achieving the optimal value can always be chosen to be a quantum state. Nonetheless, we expect  this to be the case, which leads us to formulate the following conjecture.
\begin{conjecture}
If $\vec{p}\in\mathcal{Q}$, then there exists $\rho\in\mathbb{D}$ such that $\rho\rightarrow\vec{p}$ and $r_\mathbb{S}(\rho)=r_\mathcal{L}(\vec{p})$.
\end{conjecture}

\subsection{Device-independent bounds on entanglement}

An approach to the quantification of the content of non-local correlations of a quantum state similar in spirit to $r_\mathcal{L}(O)$ was initiated in~\cite{EPR2} and later developed by many researchers (see~\cite{gisinnonlocalcontent} and references therein). Said approach is based on a decomposition of $\vec{p}$ different than the one resulting from the evaluation of the robustness $r_\mathcal{L}(\vec{p})$. More in detail, consider a convex decomposition
\beq
\label{eq:bestlocalapprox}
\vec{p}=(1-q_{\textrm{NL}})\vec{p}_\cL+q_{\textrm{NL}}\vec{p}_{\mathcal{NS}}.
\eeq
Here $0\leq q_{\textrm{NL}} \leq 1$ is the \emph{non-local weight} in the decomposition (\ref{eq:bestlocalapprox}), with $\vec{p}_\cL\in\mathcal{L}$ and $\vec{p}_\mathcal{NS}\in \mathcal{NS}$. One can further define $q^\textrm{min}_{\textrm{NL}}(\vec{p})$ as the minimum non-local weight over all possible decompositions (\ref{eq:bestlocalapprox}). The decomposition corresponding to the latter choice is the \emph{best local approximation} to $\vec{p}$. The approach of~\cite{EPR2} is then that of quantifying the non-local content of a quantum state $\rho$ as the maximum $q^\textrm{min}_{\textrm{NL}}(\vec{p})$ over all correlations that can be obtained from $\rho$ by local measurements.

We remark that the best local approximation has a correspondent at the level of operators in the \emph{best separable approximation} for quantum states~\cite{bestseparableapprox}. For the latter, consider decompositions
\beq
\label{eq:bestsepapprox}
\rho=(1-q_\textrm{E})\sigma_{\mathbb{S}} + q_\textrm{E} \rho_\textrm{E},
\eeq
with $0\leq q_\textrm{E}\leq 1$, $\sigma_{\mathbb{S}}\in \mathbb{S}$ and $\rho_\textrm{E}\in\mathbb{D}$. The best separable approximation corresponds to the case where $q_\textrm{E}$ is minimized, giving raise to a parameter $q^\textrm{min}_\textrm{E}(\rho)$.

A key aspect of both decompositions (\ref{eq:bestlocalapprox}) and (\ref{eq:bestsepapprox}) is that the second term on the right-hand side is constrained to be positive, i.e., a \emph{bona fide} probability distribution and a positive semidefinite operator, respectively. This prevents us from choosing $q_{\textrm{NL}}$ and $q_\textrm{E}$ arbitrarily small. One then realizes that an approach similar to \eqref{eq:bestsepapprox} does not have an immediate correspondent for pseudo-states. That is, when we write, for example, $O=(1-q)\sigma_{\mathbb{S}} + q O'$, with $\sigma_{\mathbb{S}}\in \mathbb{S}$ and $O'$ another pseudo-state, we can always choose $q$ arbitrary small, by a suitable choice of $O'$. Notice that we would encounter the same problem if the ``noise'' in \eqref{eq:entanglementrobustness} and in \eqref{eq:pseudorobustness} were allowed to be arbitrary, i.e., a pseudo-state, rather than separable or positive-semidefinite, respectively. Nonetheless it is possible to define consistently a \emph{generalized robustness of entanglement}, where one considers the minimum mixing with any arbitrary---i.e., potentially entangled---state $\rho_\mathbb{D}$ so that the resulting mixed operator is separable:
\beq
\label{eq:genralizeentanglementrobustness}r^\textrm{G}_\mathbb{S}(O):=\min\left\{t \left|\; t\geq0, \frac{O+t \rho_\mathbb{D}}{1+t}\in \mathbb{S} \textrm{ for some } \rho_\mathbb{D} \in \mathbb{D} \right\}\right..\\
\eeq
This quantity was introduced and studied in~\cite{generalizedrobustnesssteiner,generalizedrobustnessharrow} for quantum states, but one can consider it for pseudo-states as well, as we do. In the case of states, $r^\textrm{G}_\mathbb{S}$ is an entanglement measure with operational meaning~\cite{brandaodatta,brandaorobustness}.  Similarly, one can consider a \emph{generalized locality robustness} for correlations:
\[
\label{eq:genralizelocalityrobustness}
r^\textrm{G}_\mathcal{L}(\vec{p}):=\min\left\{t \left|\; t\geq0, \frac{\vec{p}+t \vec{p}_\mathcal{NL}}{1+t}\in \mathcal{L} \textrm{ for some } \vec{p}_\mathcal{NL} \in \mathcal{NL} \right\}\right..
\]
Clearly, $r^\textrm{G}_\mathbb{S}(O)\leq r_\mathbb{S}(O)$ and $r^\textrm{G}_\mathcal{L}(\vec{p})\leq r_\mathcal{L}(\vec{p})$.


Suppose that we are interested in assessing the entanglement of a distributed state $\rho=\rho_{AB}$ that is locally measured and leads to the establishment of correlations $\vec{p}$. We notice that the inequality $r^\textrm{G}_\mathcal{L}(\vec{p})\leq r^\textrm{G}_\mathbb{S}(\rho)$ ($\leq r_\mathbb{S}(\rho)$) can be proved along similar lines as $r_\mathcal{L}(\vec{p})\leq r_\mathbb{S}(\rho)$ (see paragraph before Lemma~\ref{lem:Oadd}~\footnote{Equality cannot be proven as in Proposition~\ref{prop:equalityrobustness}, because of the constraint on the noise used in the definition \eqref{eq:genralizeentanglementrobustness} of  $r^\textrm{G}_\mathbb{S}$.}), and can be used to lower bound the robustness and generalized robustness of the (potentially unknown) underlying $\rho$. This lower bound does not depend on the details of the local measurements. Thus, it constitutes a \emph{device-independent} bound on the entanglement of the underlying state. Similarly, the best local approximation of $\vec{p}$ and the best separable approximation of $\rho$ satisfy $q^\textrm{min}_{\textrm{NL}}(\vec{p})\leq  q^\textrm{min}_\textrm{E}(\rho)$. The latter fact is more or less explicitly remarked in, e.g.,~\cite{gisinnonlocalcontent}.

While a detailed study of this kind of device-independent bounds on entanglement will be reported elsewhere~\cite{pianiinprep}, we remark here that the calculation of $r^\textrm{G}_\mathcal{L}(\vec{p})$ corresponds to the solution of a linear-programming problem, hence it is simple and straightforward, at least numerically. It is also as faithful as possible for a device-independent bound: any measured non-local $\vec{p}$ allows us to find a non-trivial lower bound for the operationally meaningful $r^\textrm{G}_\mathbb{S}(\rho)$.

\subsection{Local processing of correlations with side non-resources}
\label{sec:localprocessing}

The robustness quantifiers of correlations, as well as any other possible quantifier of correlations, are interesting and potentially useful, but their meaningfulness is not immediately apparent.  The point is that they should at least behave meaningfully in an operational framework. This is the case, for example, of entanglement measures~\cite{entanglementreview}. Our approach here is to define the notion of local processing of a box with side ``non-resources''. Analyzing  such processing is made simpler by working within the unified framework for correlations summarized by (\ref{eq:oformalism}).

As represented in Figure~\ref{fig:manipulation_box}, a new box $\vec{p}'$ (with inputs labelled $x',y'$ and outputs labelled $a',b'$) can be ``built'' out of the box $\vec{p}$, processing the latter with the help of a shared (pseudo-)state $O_{\textrm{NR}}$ and local measurements. Here the subscript $\textrm{NR}$ stands for the fact that $O_{\textrm{NR}}$ is a \emph{non-resource}, meaning that via local measurements on it alone---that is, via (\ref{eq:oformalism})---it would give raise to non-resource boxes with respect to which $\vec{p}$ is compared. For example, to study the non-locality of $\vec{p}$, i.e., to compare it to local boxes, $O_{\textrm{NR}}$ is chosen to be a separable state. Similarly, to discuss how beyond quantum $\vec{p}$ is, $O_{\textrm{NR}}$ is chosen to be a \emph{bona fide} quantum state, i.e., positive semidefinite.
\begin{figure}
\begin{center}
\includegraphics[width=0.7\textwidth]{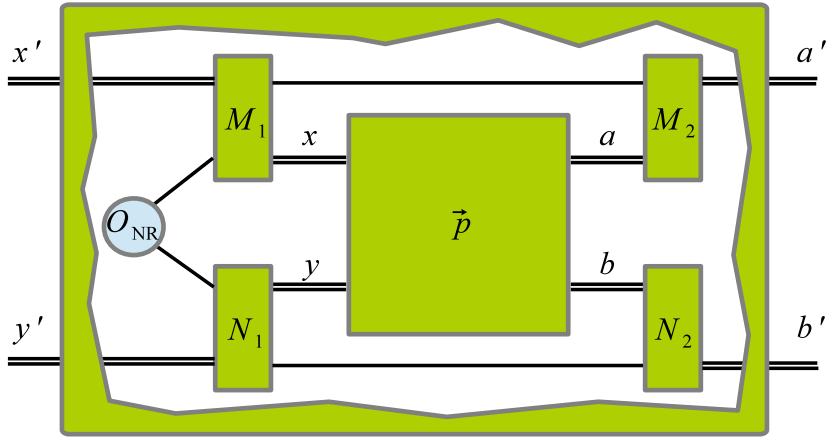}
\caption{Local processing of a box $\vec{p}$ with inputs $(x,y)$ and outputs $(a,b)$. Besides local measurements, we allow the use of any side non-resource embodied by a shared $O_\textrm{NR}$. The result is a new box $\vec{p}'$ with with inputs $(x',y')$ and outputs $(a',b')$. Classical information flows from left to right along double lines; quantum information flows along single lines.}
\label{fig:manipulation_box}
\end{center}
\end{figure}
Notice that two measurements per party are allowed: a pre-processing one (e.g., $M_1$ for Alice) before providing the input for the box $\vec{p}$, and a post-processing one (e.g., $M_2$ for Alice) using the output of $\vec{p}$. Most importantly, $M_1$ has both classical and quantum input/output, and $M_2$, albeit having only classical output, has both classical and quantum input. Information between the first measurement and the second measurement ``travels'' not only through $\vec{p}$ (which likely modifies it), but also along a quantum wire, which can accommodate both quantum and classical information. So, for example, information about Alice's classical input $x$ to $\vec{p}$, i.e., the classical output of measurement $M_1$, can be imagined to be available also at the time of the second measurement $M_2$. The quantum wires allow us to consider any non-resource arising---via local measurements---from a non-resource (pseudo-)state as being present from the very beginning, and incorporated in $O_{\textrm{NR}}$. This is the first simplification in the analysis due to the two-way relation $(\ref{eq:oformalism})$.

A second simplification becomes apparent when the box $\vec{p}$ is ``opened'' (see Figure~\ref{fig:manipulation_open_box}).
\begin{figure}
\begin{center}
\includegraphics[width=0.7\textwidth]{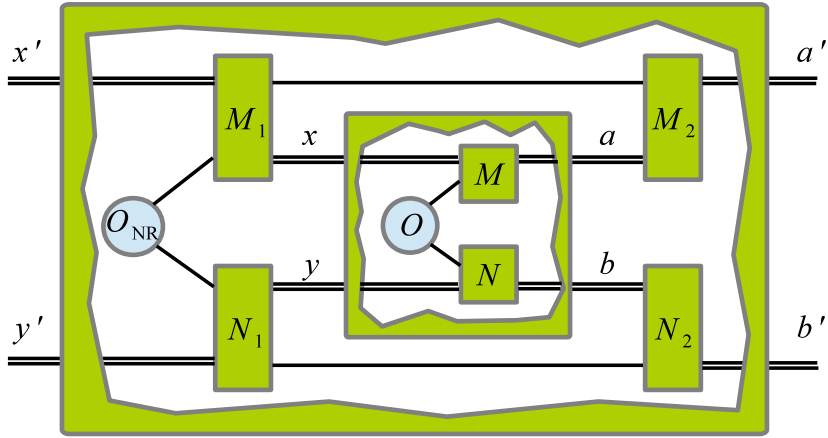}
\caption{Analysis of the local processing of a box $\vec{p}$ with side non-resources taking full advantage of the $O$-formalism: $\vec{p}$ itself arises from local quantum measurements, even if not quantum.}
\label{fig:manipulation_open_box}
\end{center}
\end{figure}
One realizes that each local processing---on Alice's side and, independently and similarly, on Bob's side---of ``wired'' (both classically and quantumly) measurements on $O_{\textrm{NR}}$ and $O$ corresponds to a two-way LOCC (local operations and classical communication) measurement on $O_{\textrm{NR}}\otimes O$. More in detail, let $O=O^{AB}$ and $O_{\textrm{NR}}=O_{\textrm{NR}}^{CD}$, with $AC$ on Alice's side, and $BD$ on Bob's side. Then one has
\begin{eqnarray}
p'(a',b'|x'y') &=\sum_{a,b,x,y} q(a',b',x,y|a,b,x',y')p(a,b|x,y)\\
&=
 \sum_{a,b,x,y} \Tr_{CD}
\Bigg(M^C_{2,a'|a}\otimes N^D_{2,b'|b} \left(\Lambda^C_{x|x'}\otimes \Gamma^D_{y|y'} \left[O^{CD}_{\textrm{NR}}\right]\right)\Bigg)\\
&\qquad\quad\times\Tr_{AB}\left(M^A_{a|x}\otimes N^B_{b|y} O^{AB}\right)\\
%
&=\Tr\Bigg(
\bigg(
	\sum_{a,x} \Lambda^{C\dagger}_{x|x'}[M^C_{2,a'|a}]\otimes M^A_{a|x}
\bigg)\\
&\qquad\otimes	
\bigg(
	\sum_{b,y} \Gamma^{D\dagger}_{y|y'}[N^D_{2,b'|b}]\otimes N^A_{b|y}
\bigg)\,\,
O^{CD}_{\textrm{NR}}\otimes O^{AB}
\Bigg)\\
&=\Tr\big(M'^{CA}_{a'|x'}\otimes N'^{DB}_{b'|y'}\,O^{CD}_{\textrm{NR}}\otimes O^{AB}\big)
\end{eqnarray}
Here the completely-positive trace-non-increasing maps $\Lambda^C_{x|x'}$ (such that $\sum_{x}\Lambda^C_{x|x'}$ is trace-preserving for all $x'$) correspond to the action of measurement $M_1$ (similarly for the $\Gamma$'s and $N_1$). By using their duals $\Lambda^{C\dagger}_{x|x'}$ we defined new local POVM elements $M'^{CA}_{a'|x'}=\sum_{a,x} \Lambda^{C\dagger}_{x|x'}[M^C_{2,a'|a}]\otimes M^A_{a|x}$ that can be realized via two-way LOCC between $A$ and $C$ (similarly on Bob's side). One concludes that a box $\vec{p}$ that can be realized via local measurements on some $O$ gets mapped onto a box $\vec{p'}$ that can be realized via local measurements on $O\otimes O_{\textrm{NR}}$.

Consider then the case where one quantifies a property of boxes in terms of the underlying (pseudo-)states. Suppose one proves that such a quantifier behaves well---in particular, that it does not increase---under tensoring the (pseudo-)state with a non-resource (pseudo-)state; then it is proven that the quantifier behaves meaningfully under local processes with side non-resources.

We are going to show that this is the case for both $r_\mathcal{L}(\vec{p})= r_\mathbb{S}(\vec{p})$ and $r_\mathcal{Q}(\vec{p})= r_\mathbb{D}(\vec{p})$ (remember that these equalities were proven in Proposition~\ref{prop:equalityrobustness}). By making use of the unified framework for correlations we will now explicitly argue that $r_\mathcal{L}(\vec{p})$ behaves well under local processing with side non-resources. The same kind of argument goes through for $r_\mathcal{Q}(\vec{p})$.

Consider an optimal $O$ such that $O\rightarrow\vec{p}$ and $r_\mathbb{S}(\vec{p})=r_\mathbb{S}(O)$. Consider further a separable $\sigma_\mathbb{S}$ such that $(O+r_\mathbb{S}(O)\sigma_\mathbb{S})/(1+r_\mathbb{S}(O))$ is separable, which exists by definition of $r_\mathbb{S}(O)$. A box $\vec{p'}$ that we can obtain with local processing from $\vec{p}$ using as side non-resource a separable state $\tau_\mathbb{S}$, can be obtained by measuring locally $O\otimes \tau_\mathbb{S}$. We have that
\[
\frac{O\otimes \tau_\mathbb{S} + r_\mathbb{S}(O) \sigma_\mathbb{S}\otimes \tau_\mathbb{S}}{1+r_\mathbb{S}(O)}=\frac{O + r_\mathbb{S}(O) \sigma_\mathbb{S}}{1+r_\mathbb{S}(O)}\otimes \tau_\mathbb{S}
\]
is separable, because the tensor product of two separable states is separable. We observe that there might be another $\rho_\mathbb{S}$ such that such that $(O\otimes \tau_\mathbb{S} + t \rho_\mathbb{S} )/(1+t)$ is separable for $0\leq t\leq r_\mathbb{S}(O)$. Thus, we have that $r_\mathbb{S}(O\otimes \tau_\mathbb{S})\leq r_\mathbb{S}(O)=r_\mathbb{S}(\vec{p})$. Since $O\otimes \tau_\mathbb{S}\rightarrow \vec{p'}$, this implies $r_\mathbb{S}(\vec{p'})\leq r_\mathbb{S}(\vec{p})$. Thus, local processing with side non-resources does not increase $r_\mathcal{L}(\vec{p})= r_\mathbb{S}(\vec{p})$.

\subsection{Alternative characterizations of robustness}

We provide here two alternative characterizations of robustness that will apply to our case.
\begin{proposition}
The robustness (\ref{eq:robustness}) has  the two alternative characterizations
\begin{eqnarray}
\label{eq:robustsum}r_S(v)&=\frac{1}{2}\left(\min\left\{\sum_i |c_i|\left| v = \sum_i c_i w_i, w_i\in S, \sum_i c_i =1\right\}\right.-1\right)\\
\label{eq:robustfunction}	   &=\frac{1}{2}\bigg(\max\Big\{\,|f(v)|\,\Big|\, |f(w)|\leq 1\,\forall w\in S\Big\}-1\bigg),
\end{eqnarray}
where in (\ref{eq:robustfunction}) we consider linear real functionals on the real affine space $\mathcal{A}$ of which $S$ is a subset.
\end{proposition}
\begin{proof}
For completeness we will prove (\ref{eq:robustsum}). A similar proof appears in \cite{communicationcomplexity}; in the same reference the reader will find the less straightforward proof of (\ref{eq:robustfunction}).

That $r_S(v)$ is greater than the right-hand side of (\ref{eq:robustsum}) is clear, because we can write $v=(1+r_S(v))w^+-r_S(v)w^-$. On the other hand, consider an optimal decomposition $v=\sum_i c_i w_i$, achieving the minimum on the right-hand side of (\ref{eq:robustsum}). We can write
\begin{eqnarray*}
v&=\sum_i c_i w_i\\
  &= \sum_{c_i\geq 0} c_i w_i + \sum_{c_i< 0} c_i w_i \\
  &= \sum_{c_i\geq 0} |c_i| \frac{\sum_{c_i\geq 0} |c_i| w_i}{\sum_{c_i\geq 0} |c_i| } - \sum_{c_i< 0}|c_i|\frac{\sum_{c_i< 0}|c_i| w_i }{\sum_{c_i< 0}|c_i|}\\
  &= \left(1+\sum_{c_i< 0}|c_i|\right) w^+ - \left(\sum_{c_i< 0}|c_i|\right) w^-,
\end{eqnarray*}
having defined $w^+=\frac{\sum_{c_i\geq 0} |c_i| w_i}{\sum_{c_i\geq 0} |c_i| }$ and $w^-=\frac{\sum_{c_i< 0}|c_i| w_i }{\sum_{c_i< 0}|c_i|}$ and taken into account that $\sum_ic_i= \sum_{c_i\geq 0} |c_i| - \sum_{c_i< 0} |c_i| =1$. So, we see that $r_S(v)\leq \sum_{c_i< 0} |c_i|= (\sum_i |c_i|-1)/2$.
\end{proof}

This implies the following:
\begin{eqnarray}
\phantom{\textrm{with }}\quad  r_\mathcal{L}(\vec{p})&=& (b_\cL(\vec{p})-1)/2\\
\textrm{with }\quad  b_\cL(\vec{p})&:=&\min\Bigg\{\sum_i |c_i| \Bigg| \,\vec{p} = \sum_i c_i \vec{p^i}_\mathcal{L}, \vec{p^i}_\mathcal{L}\in \mathcal{L}, \sum_i c_i =1\Bigg\}\\
				&=& \max\Bigg\{\,|B_\mathcal{L}(\vec{p})|\,\Bigg|\, |B_\mathcal{L}(\vec{p}_\mathcal{L})|\leq 1\,\forall \vec{p}_\mathcal{L}\in \mathcal{L}\Bigg\},\\[8pt]
\label{eq:Qrobustness}\phantom{\textrm{with }}\quad r_\mathcal{Q}(\vec{p}) &=& (b_\cQ(\vec{p})-1)/2\\
\textrm{with }\quad b_\cQ(\vec{p}) &:=&\min\left\{\sum_i |c_i|\left| \,\vec{p} = \sum_i c_i \vec{p^i}_\mathcal{Q}, \vec{p^i}_\mathcal{Q}\in \mathcal{L}, \sum_i c_i =1\right\}\right.\\
				&=& \max\Big\{\,|B_\mathcal{Q}(\vec{p})|\,\Big|\, |B_\mathcal{Q}(\vec{p}_\mathcal{Q})|\leq 1\,\forall \vec{p}_\mathcal{Q}\in \mathcal{Q}\Big\},\\[8pt]
\phantom{\textrm{with }}\quad r_\mathbb{S}(O) &=& (w_\mathbb{S}(O)-1)/2\\
\textrm{with }\quad w_\mathbb{S}(O)&:=& \min\left\{\sum_i |c_i|\left| \,O = \sum_i c_i \sigma_i, \sigma_i\in \mathbb{S}, \sum_i c_i =1\right\}\right.\\
				&=& \max\Big\{\,|\Tr(WO)|\,\Big|\, |\Tr(W\sigma)|\leq 1\,\forall \sigma\in\mathbb{S}\Big\},\\[8pt]
\phantom{\textrm{with }}\quad r_\mathbb{D}(O) &=&  (w_\mathbb{D}(O)-1)/2\\
\textrm{with }\quad w_\mathbb{D}(O)&:=&\min\left\{\sum_i |c_i|\left| \,O = \sum_i c_i \rho_i, \rho_i\in \mathbb{D}, \sum_i c_i =1\right\}\right.\\
				&=& \max\Big\{\,|\Tr(WO)|\,\Big|\, |\Tr(W\rho)|\leq 1\,\forall \rho\in\mathbb{D}\Big\}.
\end{eqnarray}
In the above we have used the notation $B_{\mathcal{L}/\mathcal{Q}}(\vec{p})=\sum_{a,b,x,y}B_{\mathcal{L}/\mathcal{Q}}(a,b|x,y)p(a,b|x,y)$, $B_{\mathcal{L}/\mathcal{Q}}(a,b|x,y)\in\mathbb{R}$, and considered operators $W=W^\dagger$. The functions $B_{\mathcal{L}/\mathcal{Q}}$ play the role of Bell parameters/values~\cite{nonlocalityreview}. The constraint $|B_\mathcal{L}(\vec{p}_\mathcal{L})|\leq 1$ can be seen as a Bell inequality~\cite{bell} bounding classical correlations, with $B_\mathcal{L}(a,b|x,y)$ the coefficients in the inequality; $|B_\mathcal{Q}(\vec{p}_\mathcal{Q})|\leq 1$ can instead be seen as a Tsirelson inequality~\cite{tsirelson}, bounding quantum correlations. On the other hand, the $W$'s play the role of witnesses. The constraint $|\Tr(W\sigma)|\leq 1$ for all $\sigma\in\mathbb{S}$ can be seen as a condition for an entanglement witness~\cite{entanglementreview}, normalized and rescaled differently than usual, as typically it is asked that $\Tr(W\sigma)\geq 0$ for all $\sigma\in\mathbb{S}$, with potentially $\Tr(W\rho)<0$ for some entangled state $\rho$. The request  $|\Tr(W\rho)|\leq 1\,\forall \rho\in\mathbb{D}$ can instead be seen as a condition for a witness of lack of positivity; also here, the approach is different than usual, as in order to detect non-positivily one would typically consider $W$'s that are positive operators themselves and check the condition $\Tr(W\rho)\geq 0$. Notice that because of linearity, we could have simply used a unified notation based on the notion of inner product (e.g., $\Tr(WO)$ is nothing else than the Hilbert-Schmidt inner product between $W$ and $O$).
\section{Quantifying beyond-quantum correlations in the unified operator formalism}

In the following we will focus on making use of the unified framework for correlations in the analysis and quantification of correlations beyond quantum. In particular, we will find alternative expressions for $r_\mathbb{D}(\vec{p})=r_\mathcal{Q}(\vec{p})$.

Let us start by introducing some notation. The trace norm of an operator $X$ is defined as $\|X\|_1:=\Tr\sqrt{X^\dagger X}$, i.e., as the sum $\sum_i\sigma_i(X)$ of all the singular values of $X$. If $X$ is Hermitian, $X=X^\dagger=\sum_i x_i \ket{x_i}\bra{x_i}$, with $x_i\in\mathbb{R}$ the eigenvalues of $X$ and $\{\ket{x_i}\}$ its eigenbasis, then $\|X\|=\sum_i |x_i|$.

Given two Hermitian operators $X$ and $Y$, we define their trace-norm distance as
\[
D(X,Y):=\frac{1}{2}\|X-Y\|_1,
\]
where we included a normalizing factor $1/2$ such that $D(\rho,\sigma)=1$ for orthogonal states $\rho$ and $\sigma$.

Notice that, since $\Tr O =1 $ for all (pseudo-)states, $(\|O\|_1-1)/2$ corresponds to the sum of the absolute values of the negative eigenvalues of $O$. Thus, it is a clearcut quantifier of how non-positive $O$ is~\cite{negativity}. Similarly $\min_{\rho\in\mathbb{D}} D(O,\rho)$ quantifies how different from a quantum state $O$ is. We find the following relations,  which, with the exception of the interpretation in terms of distance, were derived in~\cite{negativity} in a different context---the quantification of entanglement by means of partial transpositon and \emph{negativity}.
\begin{lemma}
\label{lem:negativity}
It holds
\[
r_\mathbb{D}(O)=\min_\rho D(O,\rho) = \frac{\|O\|_1 - 1}{2}.
\]
\end{lemma}
\begin{proof}
Any Hermitian $O=\sum_i o_i \ket{o_i}\bra{o_i}$ admits a Jordan decomposition into its positive and negative parts, i.e.,
\[
O=O_+ - O_-,\quad O_+ = \sum_{o_i\geq 0} o_i \ket{o_i}\bra{o_i},\quad O_- = \sum_{o_i< 0} |o_i| \ket{o_i}\bra{o_i}.
\]
So we can write
\[
O = \Tr O_+ \rho_+ - \Tr O_-\rho_,
\]
with $\rho_\pm= O_\pm/\Tr O_\pm$. Notice that from $\|O\|_1=\Tr O_+ + \Tr O_-$ and $\Tr O_+ - \Tr O_- =\Tr O =1$ we get
$\Tr O_+ = 1+\Tr O_-$  and $\Tr O_- = (\|O\|_1-1)/2$. So $r_\mathbb{D}(O)\leq \Tr O_-= (\|O\|_1-1)/2$.

On the other hand, consider an optimal---for the sake of $r_\mathbb{D}(O)$---decomposition $O=(1+r_\mathbb{D}(O))\rho - r_\mathbb{D}(O)\rho'$, from which, using the triangle inequality, we find
\begin{eqnarray}
\nonumber \|O\|_1&=\|(1+r_\mathbb{D}(O))\rho - r_\mathbb{D}(O)\rho'\|_1\\
\nonumber			&\leq (1+r_\mathbb{D}(O))\|\rho\|_1 + r_\mathbb{D}(O)\|\rho'\|_1\\
\nonumber		&=1+2 r_\mathbb{D}(O)).
\end{eqnarray}
so that, $r_\mathbb{D}(O))\leq (\|O\|_1 - 1)/2$.

Overall, we proved $r_\mathbb{D}(O))= (\|O\|_1 - 1)/2$. To prove the remaining claim, observe that $2D(O,\rho_+)=\|O-\rho_+\|_1=\|O\|_1-\tr\rho_+=\|O\|_1 -1$. Observe also that, for any state $\rho$, $\|O-\rho\|_1\geq \|O\|_1 - \|\rho\|= \|O\|_1 - 1$. So, $\min_\rho D(O,\rho)=D(O,\rho_+)=(\|O\|_1 - 1)/2$.
\end{proof}

Thus, combining Lemma~\ref{lem:negativity} with Proposition~\ref{prop:equalityrobustness}, we arrive at
\begin{theorem}
It holds
\beq
\label{eq:Onegativity}
r_\mathbb{D}(\vec{p})=r_\mathcal{Q}(\vec{p})=\min_{O\rightarrow\vec{p}} \min_{\rho}D(O,\rho)=\min_{O\rightarrow\vec{p}}\frac{\|O\|_1 - 1}{2}.
\eeq
\end{theorem}
We remark that the rightmost-side-hand of (\ref{eq:Onegativity}) provides a very simple way of checking that $r_\mathbb{D}(\vec{p})$ does not increase under local processing with side non-resources. Indeed, let $\vec{p'}$ be a box obtained processing $\vec{p}$ with a shared quantum state $\rho$ and local measurements. Suppose $O\rightarrow \vec{p}$, with $O$ optimal for the sake of $r_\mathbb{D}(\vec{p})$. Then
\begin{eqnarray*}
r_\mathbb{D}(\vec{p'})&=\min_{O'\rightarrow\vec{p'}}\frac{\|O'\|_1 - 1}{2}\\
			&\leq \frac{\|O\otimes\rho\|_1 - 1}{2}\\
			&=\frac{\|O\otimes\|_1\|\rho\|_1 - 1}{2}\\
			&=\frac{\|O\|_1 - 1}{2}=r_\mathbb{D}(\vec{p}).
\end{eqnarray*}
The inequality is due to the fact that, as discussed in Section~\ref{sec:localprocessing}, we can obtain $\vec{p'}$ by measuring locally $O\otimes\rho$.

\subsection{Beyond quantum violation of Bell inequalities}

Suppose that the only information that we have about $\vec{p}$ is that it violates some Bell inequality beyond the extent allowed by quantum mechanics. Can we bound the ``lack of physicality''---i.e., the violation of positivity---of the underlying pseudo-state with such knowledge? We can, as follows.

\begin{proposition}
\label{pro:minneg}
For any Bell inequality $B_\mathcal{L}$ is holds
\beq
\label{eq:grothendiek1}
\frac{|B_\mathcal{L}(\vec{p})|/B_\mathcal{L}^{\mathcal{Q}_\textrm{max}}-1}{2}\leq r_\mathbb{D}(\vec{p}),
\eeq
with $B_\mathcal{L}^{\mathcal{Q}_\textrm{max}}=\max_{\vec{p}_\mathcal{Q}\in\mathcal{Q}}|B_\mathcal{L}(\vec{p}_\mathcal{Q})|$. Moreover,
\beq
\label{eq:grothendiek2}
\frac{b_\mathcal{L}(\vec{p})/b_\mathcal{L}^{\mathcal{Q}_\textrm{max}}-1}{2}\leq r_\mathbb{D}(\vec{p}),
\eeq
with $b_\mathcal{L}^{\mathcal{Q}_\textrm{max}}=\max_{\vec{p}_\mathcal{Q}\in\mathcal{Q}}b_\cL(\vec{p}_\mathcal{Q})=\max_{B_\cL}B_\mathcal{L}^{\mathcal{Q}_\textrm{max}}$.
\end{proposition}
\begin{proof}
Let $\vec{p}=\sum_i c^\cQ_i\vec{p^i}_\cQ$ be an optimal decomposition of $\vec{p}$ for the sake of $r_\mathbb{D}(\vec{p})=r_\cQ (\vec{p})$ according to (\ref{eq:Qrobustness}), i.e., $r_\cQ (\vec{p})=(\sum_i |c^\cQ_i|-1)/2=(b_\cQ(\vec{p})-1)/2$. Then
\beq
\label{eq:grothendiekproof}
|B_\cL(\vec{p})|
=
|\sum_ic^\cQ_i B_\cL(\vec{p^i}_\cQ)|
\leq
B_\mathcal{L}^{\mathcal{Q}_\textrm{max}}\sum_i|c^\cQ_i|
=
B_\mathcal{L}^{\mathcal{Q}_\textrm{max}} b_\cQ(\vec{p})
\eeq
This gives (\ref{eq:grothendiek1}); upon maximization over $B_\cL$ of the leftmost and rightmost side of (\ref{eq:grothendiekproof}) we obtain (\ref{eq:grothendiek2}).
\end{proof}
The bound~\eqref{eq:grothendiek2} was already presented in~\cite{communicationcomplexity} in terms of $r_\mathcal{Q}(\vec{p})$. The bound~\eqref{eq:grothendiek1} is useful when we have information about a specific Bell inequality; a similar bound for the maximal non-local content $q_{\textrm{NL}}^\textrm{min}$ was presented in~\cite{nonlocalcontentbell}.

\subsection{Example: noisy Popescu-Rohrilich box}
\label{sec:PR}

The noisy Popescu-Rohrlich box~\cite{masanelgeneral}, which depends on a parameter $0\leq \epsilon \leq 1/2$, is defined via
\[
p_\epsilon(a,b|x,y)
:=
\left\{
\begin{array}{cl}
\frac{1-\epsilon}{2} &\textrm{if } a\oplus b =x\cdot y\\
\frac{\epsilon}{2} &\textrm{otherwise}
\end{array} 
\right.,
\]
with $a,b,x,y\in\{0,1\}$. For $\epsilon =0$ we recover the original Popescu-Rohrlich box~\cite{PR} $p_0(a,b|x,y)=p_\textrm{PR}(a,b|x,y)=1/2\, \delta_{a\oplus b,x\cdot y}$, while for $\epsilon=1/2$ we have the totally uncorrelated uniform probability distribution $p_{1/2}(a,b|x,y)=p_\textrm{rand}(a,b|x,y)=1/4$, for all $x,y,a,b$. In general, we can write
\[
p_\epsilon(a,b|x,y)=(1-2\epsilon) p_\textrm{PR}(a,b|x,y) + 2\epsilon p_\textrm{rand} (a,b|x,y).
\]
It can be easily checked that the noisy Popescu-Rohrlich box is no-signalling for all $0\leq \epsilon \leq 1/2$. On the other hand, the noisy Popescu-Rohrlich box is known to be not quantum for $0\leq\epsilon <(2-\sqrt{2})/4$.

In \cite{acin} it was shown that the original Pospescu-Rohrlich box can be obtained via \eqref{eq:oformalism} with
\[
O_\textrm{PR} := \frac{1+\sqrt{2}}{2} \ket{\psi^+}\bra{\psi^+}+\frac{1-\sqrt{2}}{2} \ket{\psi^-}\bra{\psi^-},
\]
where $\ket{\psi^\pm}=(\ket{00}\pm\ket{11})/\sqrt{2}$, and with the local POVMs (actually, projective measurements) $\{(\I\pm\sigma_x)/2\}$ and $\{(\I\pm\sigma_y)/2\}$ for Alice, and $\{(\I\pm(\sigma_x-\sigma_y)/\sqrt{2})/2\}$ and $\{(\I\pm(\sigma_x+\sigma_y)/\sqrt{2})/2\}$ for Bob. One checks that, more in general, he noisy Popescu-Rohrlich box can be obtained from
\beq
\label{eq:noisyO}
O_\epsilon:=\frac{1+\sqrt{2}(1-2\epsilon)}{2} \ket{\psi^+}\bra{\psi^+}+\frac{1-\sqrt{2}(1-2\epsilon)}{2} \ket{\psi^-}\bra{\psi^-},
\eeq
 with the same local projective measurements. The question about $O_\textrm{PR}$ that was not addressed  in~\cite{acin}, and that we will answer more in general for $O_\epsilon$,  is whether these operators violate the condition of positivity minimally, just enough to realize  $\vec{p}_\epsilon$. We answer this in the affirmative, making use of the result of Proposition~\ref{pro:minneg}. Notice that the bounds \eqref{eq:grothendiek1} and \eqref{eq:grothendiek2} are bounds on the trace norm of any pseudo-state $O$ such that $O\rightarrow\vec{p}$. More explicitly, they can be cast as
\beq
\label{eq:tsirelsonO}
|B_\mathcal{L}(\vec{p})|/B_\mathcal{L}^{\mathcal{Q}_\textrm{max}} \leq \|O\|_1
\eeq
and
\[
b_\mathcal{L}(\vec{p})/b_\mathcal{L}^{\mathcal{Q}_\textrm{max}}\leq \|O\|_1,
\]
respectively. As Bell inequality, i.e., function $B_\mathcal{L}$, we will use the CHSH inequality~\cite{CHSH}
\[
B_{\mathcal{L}, \textrm{CHSH}}(\vec{p})= \frac{1}{2}\sum_{a,b,x,y=0}^1(-1)^{(a\oplus b)\oplus (x\cdot y)}p(a,b|x,y),
\]
where the normalization factor on the right-hand side is chosen so that $|B_{\mathcal{L}, \textrm{CHSH}}(\vec{p}_\cL)|\leq1$ for all $\vec{p}_\cL\in \cL$.
One finds $B_{\mathcal{L},\textrm{CHSH}}(\vec{p}_\epsilon)=2(1-2\epsilon)$, while it is known that $B^{\mathcal{Q}_\textrm{max}}_{\mathcal{L},\textrm{CHSH}}=\sqrt{2}$. Thus, from~\eqref{eq:tsirelsonO} we find that any $O$ such that $O\rightarrow\vec{p}_\epsilon$ satisfies $\|O\|_1\geq\sqrt{2}(1-2\epsilon)$. On the other hand, $O_\epsilon$ saturates such a bound, as from \eqref{eq:noisyO} one finds exactly $\|O_\epsilon\|_1=\sqrt{2}(1-2\epsilon)$, for $0\leq \epsilon\leq (2-\sqrt{2})/4$. Notice that the closest physical state to $O_\epsilon$ in that range of $\epsilon$ is always $\ket{\psi^+}\bra{\psi^+}$, with $O_{\epsilon=(2-\sqrt{2})/4}$ exactly equal to $\ket{\psi^+}\bra{\psi^+}$.
%
%

\section{Conclusions}

In this article we have illustrated how the unified framework for no-signalling correlations introduced in~\cite{acin} can be used in the quantification of correlations. The key aspect of the unified framework is that also beyond-quantum correlations have an ``almost quantum'' representation: they are still obtained via local quantum measurements, but the distant parties share pseudo-states, i.e., trace-one Hermitian operators that are not necessarily positive semidefinite. This opens up the possibility of, for example, casting the processing of all no-signalling correlations---including beyond-quantum ones that are usually represented solely as boxes with inputs/outputs---in terms of the processing of the underlying (pseudo-)state. Here, with in mind a resource-theory approach, we have looked at the simple example of local processing with side non-resources (e.g., side separable states, when we want to study as a resource correlations arising from entanglement). It would be interesting to  consider the unified operator formalism in more sophisticated processing scenarios~(see~\cite{devicente} and references therein).

The unified framework for correlations also allows one to connect  tightly the quantification of properties of correlations intended in the usual sense of probability distributions---how non-local, how beyond-quantum---to the quantification of properties of the underlying operator---how separable, how non-positive. We showed that, when using the notion of robustness of a certain property, this connection leads actually to an identification, e.g., the robustness of non-locality is the same as the entanglement robustness. Such a connection constitutes also a simple approach to the device-independent quantification of the properties, like entanglement, of the underlying states. The application and expension of our results to other measures of distinguishability for operators and probabilities---like relative entropy---will be explored elsewhere.

We focused in particular on the quantification of how beyond-quantum certain correlations are. We expect this to be useful in the quest to pin down---via some physical or information-theoretic principle---quantum correlations among all no-signallng correlations. One provocative way to cast such quest is that of asking how comes pseudo-states are not allowed. This is particularly relevant in the multipartite scenario which we have not considered here, as it was proven that there are pseudo-states that always lead to bona-fide probability distributions, even without constraints on the allowed measurements~\cite{acin}.

As a natural quantifier of beyond-quantum correlations, we have analyzed the connection between the beyond-quantum robustness of pseudo-states and of probability distributions. As already observed in~\cite{negativity}, the first coincides with the negativity of the pseudo-state. As an example of the concepts developed, we have looked at the specific case of the family of noisy Popescu-Rohrlich boxes, proving the optimality---in terms of negativity---of the underlying pseudo-state considered in~\cite{acin} for the case of the perfect Popescu-Rohrlich box, and here generalized to its noisy version.

\emph{Note added.} During the completion of this manuscript we became aware of the related work by De Vicente~\cite{devicente}.

\section*{Acknowledgements}

M. P. thanks A. Ac{\`i}n and D. Cavalcanti for useful discussions. This work has been supported by NSERC, CIFAR, and Ontario Centres of Excellence.

\end{document}